\documentclass[aps,pra,reprint,groupedaddress, superscriptaddress]{revtex4-1}
\usepackage[dvips]{graphicx}
\usepackage{amsthm}
\usepackage{amsmath}
\usepackage{physics}
\usepackage{color}
\usepackage{amssymb}
\usepackage{dcolumn}
\usepackage{amsfonts}
\usepackage{bbold}
\usepackage{ulem}

\newtheorem{lemma}{Lemma}

\allowdisplaybreaks[0]

\begin{document}
\title{Generalized energy measurements and quantum work compatible with\\ fluctuation theorems}

\author{Kosuke Ito}\email{kosuke@zju.edu.cn}\affiliation{Department of Physics and Zhejiang Institute of Modern Physics, Zhejiang University, Hangzhou, Zhejiang 310027, China}
\author{Peter Talkner}\affiliation{Institut f\"ur Physik, Universit\"at Augsburg, Universit\"atsstra{\ss}e 1, 86135 Augsburg, Germany}
\author{B.~Prasanna Venkatesh}\affiliation{Indian Institute of Technology, Gandhinagar, Gujarat 382355, India}
\author{Gentaro Watanabe}\email{gentaro@zju.edu.cn}\affiliation{Department of Physics and Zhejiang Institute of Modern Physics, Zhejiang University, Hangzhou, Zhejiang 310027, China}\affiliation{Zhejiang Province Key Laboratory of Quantum Technology and Device, Zhejiang University, Hangzhou, Zhejiang 310027, China}

 \begin{abstract}
  The probability densities of work that can be exerted on a quantum system initially staying in thermal equilibrium are constrained by the fluctuation relations of Jarzynski and Crooks, when the work is determined by two projective energy measurements. We investigate the question whether these fluctuation relations may still hold if one employs generalized energy measurements rather than projective ones. Restricting ourselves to a class of universal measurements which are independent of several details of the system on which the work is done, we find sets of necessary and sufficient conditions for the Jarzynski equality and the Crooks relation. The Jarzynski equality requires perfect accuracy for the initial measurement, while the final one can be erroneous. On the other hand, the Crooks relation can only tolerate a depolarizing channel as a deviation from the projective measurement for systems with a finite dimensional Hilbert space. For a separable infinite-dimensional space only projective measurements are compatible with the Crooks relation. The results we have obtained significantly extend those of [Venkatesh, Watanabe, and Talkner, New J. Phys. {\bf 16}, 015032 (2014)] as well as avoid some errors present there.
 \end{abstract}
\maketitle

\section{Introduction}
Work in quantum mechanics, quantum work, has turned out to be a surprisingly non-trivial notion unlike its elementary character in classical mechanics. Presently, there is even no general agreement about how to measure and how to calculate quantum work. As a consequence, a variety of different definitions of quantum work exist, many of which are listed and compared in \cite{Baumer}. 

Here we restrict ourselves to thermally closed systems on which work is performed by a change of one or several external parameters $\lambda(t)$ of the Hamiltonian $H\big (\lambda(t) \big )$ of the system. The parameter change is externally controlled during the time-interval $0\le t \le \tau$ according to a prescribed force protocol $\Lambda = \{ \lambda(t)| 0 \le t \le \tau \}$. We shall further consider only the so-called two-energy measurement schemes (TEMS) consisting of energy measurements immediately before and after the force protocol at $t=0$ and $\tau$, respectively \cite{CHT}. The work performed on the system in a particular realization of the force protocol is given by the difference between the measured initial and final energy. This work is a random quantity described by a probability density function (pdf) $p_\Lambda(w)$ to find the amount of work $w$ for the protocol $\Lambda$. Under the additional conditions that the energy measurements are projective and the initial state is Gibbsian, the average exponentiated work fulfills the Jarzynski equality in the same form as for classical systems~\cite{Jarzynski,Kurchan,Tasaki,TLH}
\begin{equation}
\langle e^{-\beta w} \rangle_\Lambda = e^{-\beta \Delta F}\:,
\label{JE}
\end{equation}
where $\beta$ is the inverse temperature characterizing the initial canonical state 
\begin{equation}
\rho_\beta(\lambda(0)) = Z^{-1}(0)\: e^{- \beta H (\lambda(0))}\:. 
\label{can}
\end{equation}
Here, the free energy difference $\Delta F =F(\lambda(\tau)) -F(\lambda(0))$ referring to the initial and final parameter values results from the corresponding partition functions $Z(t) = \Tr e^{-\beta H (\lambda(t))} = e^{-\beta F(\lambda(t))}$.
Under the above conditions, the Crooks relation~\cite{Tasaki,TH2007,Crooks} also assumes the same form as in the classical case, hence, reading
\begin{equation}
p_\Lambda(w)= e^{-\beta(\Delta F -w)} p_{\bar{\Lambda}}(-w)\:.
\label{CR}
\end{equation}
It relates the pdf $p_\Lambda(w)$ for the protocol $\Lambda =\{\lambda(t)|0 \le t \le \tau \}$ to the work pdf $p_{\bar{\Lambda}}$ for the time-reversed protocol $\bar{\Lambda} = \{\epsilon_\lambda \lambda(\tau -t)| 0 \le t \le \tau \}$, where $\epsilon_\lambda$ denotes the parity of the parameter $\lambda$ under time-reversal. The underlying Hamiltonians conform by assumption with the time-reversal symmetry expressed as $\theta^\dagger H\big (\lambda(t) \big ) \theta = H\big (\epsilon_\lambda \lambda(t) \big )$, where $\theta$ is the time-reversal operator~\cite{Messiah}.

From the Crooks relation, the Jarzynski equality (\ref{JE}) for the forward process as well as the analogous equality for the backward process follows. The latter reads
\begin{equation}
\langle e^{-\beta w} \rangle_{\bar{\Lambda}} = e^{\beta \Delta F}
\label{bJE}
\end{equation}
with the same free energy difference $\Delta F$ as in the forward Jarzynski equality (\ref{JE}).
As one cannot deduce the Crooks relation from the two Jarzynski equalities, the Crooks relation is the stronger one.

Projective measurements, though they present an idealization, are difficult to realize in general. Therefore, it is natural to ask whether there is a more general class of energy measurements (so-called generalized energy measurements) for which the fluctuation theorems of Jarzynski and Crooks still hold~\cite{VWT}. The question about alternative definitions of work is also motivated by the fact that, apart from the case of stationary initial conditions, projective energy measurements erase any correlations in the energy eigenbasis and, consequently, influence the outcome of the second energy measurement and work.

In this context, we note the no-go theorem by Perarnau-Llobet {\it et al.}~\cite{Perarnau-Llobet}. It states the non-existence of an operational definition of work in terms of  generalized measurements fulfilling two requirements: (i) The average work must always agree with the difference between the average Hamiltonian at the final and the initial time. Here, both averages are taken with respect to the initial state that may contain non-diagonal matrix elements in the energy eigenbasis of the initial Hamiltonian. (ii) For initial states diagonal in the eigenbasis of the initial Hamiltonian, the distribution of work should agree with the one obtained by the two projective energy measurement scheme (TPEMS). The first condition, which was formulated earlier in~\cite{Pusz}, and has been frequently used in the study of quantum engines \cite{Alicki,Kosloff}, disregards the quantum mechanical impact of the first energy measurement on the state of the considered quantum system.
A generalized energy measurement need not lead to a complete reduction of the state with respect to the energy eigenbasis, but  its impact cannot be completely suppressed if the initial state possesses correlations in the energy basis.
The second requirement might be considered as quite restrictive in demanding perfect agreement with the work distribution following from TPEMS.        

Here we study the problem of the existence of a TEMS with generalized energy measurements for which the Crooks relation and consequently also the Jarzynski equality hold, or for which only the Jarzynski equality is satisfied. With the limitation to TEMS, the class of generalized work measurements is more restricted than the one considered in the mentioned no-go theorem in Ref.~\cite{Perarnau-Llobet}. Here we do not consider work measurement schemes for which the contact between the system and a single measurement apparatus is kept during the force protocol until  a single reading of the  measurement apparatus yields the work. Examples of such ``work-meters'' are discussed in Refs.~\cite{deChiara,HayTaj,TH2016,Cerisola}. In the present paper, we do not impose any extra condition on the average work, nor do we require the agreement of the work distribution with the form following from the TPEMS. We note that the class of generalized energy measurements we consider here within the TEMS is more general than the one considered in Ref.~\cite{VWT} as will be detailed below. Moreover, we avoid the errors contained in Ref.~\cite{VWT}.

The paper is organized as follows. Section~\ref{GM} introduces the notation we use and provides a comprehensive description of the mathematical tools needed to deal with generalized measurements. The form of the work pdf for the TEMS is reviewed in Section~\ref{sec_work}.  Our main results are presented in Section~\ref{results}. In Sections~\ref{Jarzynski} and~\ref{Crooks}, we present discussions leading to the respective conditions on the energy measurements, given in Section~\ref{results}, required for the validity of the Jarzynski equality and the Crooks relation. In Section~\ref{Conclusion} we conclude with a summary, and the appendices provide various technical details.
  
\section{Generalized energy measurements}\label{GM}
Any measurement requires a contact of the system under investigation with a measurement apparatus. To qualify as a measurement, this contact is supposed to change the state of the apparatus in a way that one can infer about the value $\omega$ of the quantity to be measured, but at the same time it also influences the state of the system. For classical systems, the resulting backaction on the system can be made arbitrarily small in principle, but this is not the case for quantum systems. Under the assumption that, prior to the measurement, the system and the apparatus are totally independent of each other, and thus initially stay in a product state, the state of the system immediately after the measurement is characterized by an operation $\phi_\omega$, which depends on the observed result $\omega$. It yields the non-normalized post-measurement state $\phi_\omega (\rho)$, where $\rho$ is the system density matrix before the measurement. As an operation, $\phi_\omega$ is a linear, completely positive (CP) map~\cite{Haag,Kraus} of the  trace-class operators $TC(\mathcal{H})$ on the Hilbert space $\mathcal{H}$ of the system, i.e., those bounded operators with a finite trace of their modulus~\cite{Schatten}. The trace of $\phi_\omega(\rho)$ determines the probability $p(\omega)$ to obtain the outcome $\omega$ of the measurement of the state $\rho$, i.e.,
\begin{equation}
\begin{split}
p(\omega) = \Tr \phi_\omega(\rho)
= \Tr \phi_\omega^*(\mathbb{1}) \rho\:.
\end{split}
\label{po}
\end{equation}    
 After the second equality sign, $p(\omega)$ is expressed in terms of the dual map $\phi_\omega^*$, which acts on bounded operators and is defined as $\Tr \phi^*_\omega(u) \rho = \Tr u \phi_\omega(\rho)$ for all bounded operators $u \in B(\mathcal{H})$ and all trace class operators $\rho \in TC(\mathcal{H})$.
The image $E_\omega=\phi^*_\omega(\mathbb{1})$ of the identity operator $\mathbb{1}$ under the dual operation is also known as the {\it effect} of the operation $\phi_\omega$~\cite{Kraus}. 
Upon normalization, the post-measurement state conditioned on the outcome $\omega$ becomes $\rho^{\text{pm}}_\omega = \phi_\omega(\rho)/ p(\omega)$.

As a CP map, the operation $\phi_\omega$ can be expressed in terms of bounded operators $B_{\omega,l}$, called Kraus operators such that~\cite{Kraus}
\begin{equation}
\phi_\omega(\rho) =\sum_l B_{\omega,l}\, \rho\, B^\dagger_{\omega,l}\:.
\label{Kr}
\end{equation}

In passing, we note that the set of operations, $\{\phi_\omega \}_{\omega \in \Omega}$, where $\Omega$ denotes the set of all possible outcomes, is referred to as an {\it instrument}~\cite{DL, Ozawa}. The operation given by the sum over $\Omega$, $\Phi = \sum_\omega \phi_\omega$, is a completely positive, trace preserving (CPTP) map describing the post-measurement state of the nonselective measurement, i.e., the one whose outcome is ignored. For further details about generalized measurements, we refer to the literature~\cite{Ozawa,Hayashi}.

In the present paper, we focus on energy measurements described by operations $\phi_n$, where $n$ indicates the eigenvalue $E_n$ of the Hamiltonian $H = \sum_n E_n \Pi_n$. Here $\{\Pi_n\}_{n \in \mathcal{N}}$ is the set of spectral projectors, satisfying $\Pi_n \Pi_m = \delta_{n,m} \Pi_n$, $\Pi_n=\Pi^\dagger_n$ and $\sum_n \Pi_n = \mathbb{1}$, where $\mathcal{N}\subseteq \mathbb{N}$.
We restrict ourselves to purely discrete spectra of the considered Hamiltonian $H$ because, otherwise, the partition function $Z = \Tr e^{-\beta H}$ and hence the free energy appearing in the fluctuation theorems does not exist.

The measurement typically contains errors that can be quantified by the probability $p(m|n)$ with which the energy $E_m$ is assigned to a state $\rho_n=\Pi_n/ \Tr \Pi_n$ with the energy $E_n$. This probability is given by
\begin{equation}
p(m|n) = \Tr \phi_m(\rho_n)\:.
\label{pmn}
\end{equation}
A measurement is error free when the error probability collapses to a Kronecker delta, i.e., $p(m|n) = \delta_{m,n}$. In this case, the operations $\phi_m$ take the following form:
\begin{equation}
\phi_m(\rho) = \mathcal{E}_m(\Pi_m\rho \Pi_m)\:,
\label{ef}
\end{equation}
where $\mathcal{E}_m$ is an arbitrary CPTP map, referred to as a quantum channel in the context of quantum information theory, depending on $m$. This means that any error-free measurement can be considered as a projective measurement followed by a trace-preserving operation.
A proof of Eq.~(\ref{ef}) is presented in Appendix~\ref{errorfree}.

For later use we mention that, as an immediate consequence of Eq.~(\ref{ef}), the domain of an error-free map can be extended from trace-class to bounded operators.
Especially, we have
\begin{equation}
\phi_n(\mathbb{1}) = \phi_n(\Pi_n)\:.
\label{phin1}
\end{equation}

  \section{Work statistics}\label{sec_work}
  As described in the introduction, we consider a thermally isolated system which undergoes the dynamics governed by a time-dependent Hamiltonian $H\big (\lambda(t)\big )$. The parameter stays constant at $\lambda(t) =\lambda(0)$ for $t\leq 0$, then changes in a prescribed way according to the  force protocol $\Lambda =\{ \lambda(t)|0\leq t \leq \tau\}$ to continue with the constant value $\lambda(t) = \lambda(\tau)$ for $t \geq \tau$.
  According to the TEMS, generalized energy measurements are performed at $t=0$ and $t=\tau$ with instruments $\{ \phi^{(0)}_n \}_{n \in \mathcal{N}_{0}}$ and $\{ \phi^{(\tau)}_n \}_{n \in \mathcal{N}_{\tau}}$ respectively, where $\mathcal{N}_0, \mathcal{N}_{\tau} \subseteq \mathbb{N}$, and $\phi^{(t)}_n$ is the measurement operation conditioned on the observation of the eigenvalue $e_n(t)$ of the Hamiltonian, which can be expressed as
\begin{equation}
H\big (\lambda(t)\big ) = \sum_{n\in \mathcal{N}_{t}} e_n(t) \Pi_n(t) \quad (t=0,\tau)\:,
\label{Ht}
\end{equation}
where $\Pi_n(t)$ is the projector on the eigenspace with energy eigenvalue $e_n(t)$. The dimension of this eigenspace is denoted by $d_n(t) = \Tr \Pi_n(t)$.
 The joint probability $p_\Lambda(m,n)$ to obtain $e_m(0)$ in the first measurement and $e_n(\tau)$ in the second measurement, is given by the trace of the non-normalized density matrix resulting from the first measurement operation, followed by the unitary time evolution from $t=0$ to $\tau$, and finally the second measurement operation:
\begin{equation}
\begin{split}
p_\Lambda(m,n) &= \Tr \phi^\tau_m (U_\Lambda \phi^0_n \big (\rho_\beta(0) \big) U^\dagger_\Lambda)\\
&= Z^{-1}(0) \sum_k p_\Lambda(m,n|k) \: d_k(0)\: e^{-\beta e_k(0)}\:.
\end{split}
\label{pLmn}
\end{equation} 
Here 
\begin{equation}
U_\Lambda = \mathcal{T} e^{-\frac{i}{\hbar}\int_0^\tau H  (\lambda(t)) dt} 
\label{UL}
\end{equation}
is the time-evolution from $t=0$ to $\tau$ with $\mathcal{T}$ denoting the chronological time ordering symbol. In the second line of Eq.~(\ref{pLmn}), $p_\Lambda(m,n|k)$ denotes the joint probability of observing the energies $e_m(\tau)$ and $e_n(0)$ conditioned on the initial state $\Pi_k(0)/d_k(0)$, where $d_k(0) = \Tr \Pi_k(0)$ is the degree of degeneracy of the eigenenergy $e_k(0)$. Thus this probability becomes
\begin{equation}
p_\Lambda(m,n|k) = \Tr \phi^\tau_m \left(U_\Lambda \phi^0_n \big (\Pi_k(0)\big )U^\dagger_\Lambda \right)/d_k(0)\:.
\label{pLmnk}
\end{equation}
 Hence, the work pdf can be expressed as
\begin{equation}
p_\Lambda(w) = \sum_{m,n} \delta \big ( w-e_m(\tau) + e_n(0) \big ) p_\Lambda(m,n)\:,
\label{pLw}
\end{equation}
where $\delta(x)$ is the Dirac $\delta$-function.
\section{Results for universal generalized measurements}\label{results}
The generalized measurements considered here are supposed to be {\it universal} in the sense that the same energy meters can be used independently of the properties of the particular system, the temperature, and the actual dynamics during the force protocol. 
In other words, the measurement operations $\phi^t_m$ with $t=0$ and $\tau$ are supposed to be independent of
\begin{enumerate}
 \item[(a)] the magnitudes of the energies $e_m(t)$,\\
 \item[(b)] the inverse temperature $\beta$ characterizing the initial state, and
 \item[(c)] the unitary dynamics $U_\Lambda$ defined in Eq.~(\ref{UL}).
\end{enumerate}

    For this class of universal measurement instruments, we find the following necessary and sufficient conditions for the Jarzynsky equality~(\ref{JE}) to hold:
    \begin{enumerate}
     \item[(JI)] For a system with a Hilbert space of finite dimension $D= \dim (\mathcal{H})$, the first measurement  must be error free and either one of the following (Ji) or (Jii) has to be satisfied:
		 \begin{enumerate}
		  \item[(Ji)] (a) the first non-selective measurement $\Phi^0 = \sum_n\phi^0_n$ is unital, i.e., $\Phi^0(\mathbb{1}) =\mathbb{1}$, and\\
		  \item[] (b) $\Tr \phi^{\tau *}_m(\mathbb{1}) = d_m(\tau)$;\\
		  \item[(Jii)] the effects $E^\tau_m = \phi^{\tau*}_m(\mathbb{1})$ of the second measurement are proportional to the identity matrix reading $E^\tau_m = [d_m(\tau)/D] \mathbb{1}$;
		 \end{enumerate}
  \item[(JII)] For systems with a separable infinite-dimensional Hilbert space, the first measurement must be error free. Further, the dual operations of the second instrument must map bounded operators onto trace-class operators such that $\Tr \phi^{\tau*}_m(\mathbb{1}) = d_m(\tau)/a$ and $\sum_n \phi^0_n(\mathbb{1}) = a \mathbb{1}$ with a positive constant $a$.
    \end{enumerate}    
Condition (Ji)(a) is equivalent to the requirement that the von Neumann entropy of the post-measurement state of the first non-selective measurement is greater than or equal to the von Neumann entropy of the state before the measurement \cite{Hayashi},
 i.e., $-\Tr \Phi^0(\rho) \ln \Phi^0(\rho) \geq -\Tr \rho \ln \rho$. 

As explained above, the Crooks relation (\ref{CR}) establishes a connection between the 
work pdf for the forward process running under the force protocol $\Lambda$ and that for the backward process under the time-reversed protocol $\bar{\Lambda}$. Both processes start from a canonical equilibrium state at the respective parameter values, at the same inverse temperature $\beta$. The measurement instruments $\phi^t_n$ and $\bar{\phi}^t_n$ of the forward and the backward processes, respectively, are also time-reversed with respect to each other, satisfying 
\begin{equation}
\bar{\phi}^t_m(\rho) =\theta^\dagger\phi^{\tau-t}_m(\theta \rho \theta^\dagger) \theta \quad (t=0,\tau)\:.
\label{tr}
\end{equation}

  For processes between non-degenerate initial and final Hamiltonians and universal measurement instruments as introduced above, we obtain the following necessary and sufficient conditions for the Crooks relation to hold:
  \begin{enumerate}
   \item[(CI)] For systems with a Hilbert space of finite dimension $D$, the measurement operations must be of the form
\begin{align}
 \hspace{30pt} \phi^t_n(\rho) = (1-\alpha) \frac{\Tr \Pi_n(t) \rho}{D}\mathbb{1} + \alpha  \Pi_n(t)  \rho  \Pi_n(t)& \nonumber\\
 (t=0,\tau)& \label{CD}
\end{align}
	       with $-(D-1)^{-1} \leq \alpha \leq 1$, which can be realized as a combination of a projective measurement and a depolarizing channel.\\
   \item[(CII)] For systems with a separable infinite-dimensional Hilbert space, only projective energy measurements are allowed.
  \end{enumerate}

For initial or final Hamiltonians with degenerate spectra, the same conditions (CI) and (CII) can be obtained with an additional universality requirement, which will be specified in Section~\ref{w_degeneracy}.

In Ref.~\cite{VWT}, only measurement operations with a single Kraus operator were considered, i.e., all the operations considered in Ref.~\cite{VWT} are restricted to  the special form $\phi^{(t)}_n(\rho) = M_n(t) \rho M^\dagger_n(t)$ with $t=0$ and $\tau$. This excludes the presence of a depolarizing channel for systems with a Hilbert space of dimension higher than two. As a condition for the Crooks relation to hold, the initial and final measurements must be error free according to Ref.~\cite{VWT}, in agreement with the results of the present work. The conclusion in Ref.~\cite{VWT} advocating the restriction to projective measurements is also in agreement with the present findings restrained to measurement operations with a single Kraus operator except for the case of two-dimensional Hilbert space, even though it is based on an incorrect argument. However, in the case of two-dimensional Hilbert space, the instrument given by Eq.~\eqref{CD} with $\alpha = -1$ can be described by a single Kraus operator as $\phi_0^{t}(\rho) = \ketbra{1;t}{0;t}\rho\ketbra{0;t}{1;t}$ and $\phi_1^{t}(\rho) = \ketbra{0;t}{1;t}\rho\ketbra{1;t}{0;t}$, where $\ket{n;t}$ is the $n$-th level energy eigenstate of $H(\lambda(t))$. This disagreement is actually caused by a mathematically incorrect argument in Ref.~\cite{VWT}. The conditions for the Jarzynski equality formulated in Ref.~\cite{VWT}, though its proof is also incorrect, corresponds to condition (JII) with $d_m(\tau)=1$ and $a=1$. In the finite-dimensional case, while the condition (Ji) with $d_m(\tau)=1$ is found, the second possibility (Jii) is missed in Ref.~\cite{VWT}.

\section{Jarzynski equality}\label{Jarzynski}
  In this section, we focus on the necessary and sufficient condition on the energy measurements for the validity of the Jarzynski equality (\ref{JE}):
  \begin{align}
   \int p_{\Lambda}(w)\: e^{-\beta w} dw = e^{-\beta \Delta F}\label{JE_11}
  \end{align}
 with $e^{-\beta \Delta F} =Z(\tau)/Z(0)$.  Using the work pdf $p_\Lambda(w)$ defined by Eq.~(\ref{pLw}) in combination with Eq.~(\ref{pLmn}), we obtain from Eq. (\ref{JE_11})  
\begin{align}
   \sum_{m, n, k} e^{-\beta(e_k(0)-e_n(0))}p_{\Lambda}(m,n|k)\hspace{1pt} d_k(0)\frac{e^{-\beta e_m(\tau)}}{Z(\tau)} = 1.\label{JE_mod}
  \end{align}
We now determine the form of those generalized energy measurement operations that are universal in the sense defined by conditions (a) -- (c) in Sec. \ref{results}.
For this purpose, we consider necessary conditions that follow from a restricted class of spectra of the initial Hamiltonian. Later we will prove the sufficiency of the conditions found in this manner. As a restriction on the energy spectrum of the initial Hamiltonian, we demand that the energy differences are not degenerate in that the differences between pairs of energies are distinct for different pairs, i.e.,
\begin{equation}
e_n(0) -e_m(0) = e_k(0)-e_l(0) \Leftrightarrow n=k \; \text{and} \; m=l.
\label{C1}
\end{equation}
Further, we consider a family of spectra satisfying this condition, parameterized by a real parameter $x$ with energy eigenvalues $e_{x,n}(0)=x e_n(0)$. For fixed $\beta>0$, the exponential factors $e^{-\beta(e_{x,k}(0) -e_{x,n}(0))}$ for different pairs of $(k, n)$ with $n \neq k$ are linearly independent functions of $x$. Hence, all the terms of the sum in Eq.~(\ref{JE_mod}) with $n\neq k$ must separately vanish, yielding
\begin{equation}
p_\Lambda(m,n|k) =0 \quad \text{for}\; n\neq k\:.
\label{pmnk0}
\end{equation}    
Hence, the first measurement is error free. Then Eq.~(\ref{JE_mod}) reads
  \begin{align}
   \sum_{m,n}p_{\Lambda}(m,n|n)\frac{d_n(0)}{d_m(\tau)} P_m = 1,
  \end{align}
  where $P_m = e^{-\beta e_m(\tau)}d_m(\tau)/Z(\tau)$ is the probability of finding $e_m(\tau)$ in the canonical state of the system for the final parameter $\lambda(\tau)$ at the inverse temperature $\beta$ of the initial state.
Because of the universality requirements (a) and (b), this condition must hold for any spectrum of the final Hamiltonian, $\{e_m(\tau)\}$, and consequently it must hold for arbitrarily chosen probabilities $P_n$ implying 
  \begin{align}
   \sum_{n}p_{\Lambda}(m,n|n)d_n(0) = d_m(\tau) \quad \text{for all} \; m.
\label{Key_JE}
  \end{align}
 Combined with Eq.~(\ref{pLmnk}), this condition becomes 
  \begin{align}
   \Tr \left[U_\Lambda^{\dagger} \phi_m^{\tau *}(\mathbb{1}) U_\Lambda \sum_{n}\phi_n^{0} (\Pi_n(0)) \right] = d_m(\tau).\label{Key_JE_0}
  \end{align}
According to the universality requirement (c), this condition must hold 
for any unitary time-evolution operator $U_\Lambda$. 
To evaluate this condition further, we first consider systems with a Hilbert space of finite dimension $D$. 
Then, according to Lemma~\ref{Lem_Jarzynski} in Appendix \ref{PoKL}, either of the following conditions must hold:
\begin{enumerate}
 \item[(Ji)] $\sum_n \phi^0_n(\mathbb{1}) = \mathbb{1}\;$  and  $\; \Tr \phi^{\tau *}_m(\mathbb{1}) =d_m(\tau)$\:. \\
 \item[(Jii)] $\phi^{\tau *}_m(\mathbb{1}) = [d_m(\tau)/D] \mathbb{1}$\:.
\end{enumerate}
To obtain (Ji) from Lemma~\ref{Lem_Jarzynski} we choose $A= \sum_n \phi^0_n(\Pi_n) = \sum_n \phi^0_n(\mathbb{1})$, where the second equality results from Eq.~(\ref{phin1}), and $B=\phi^\tau_m(\mathbb{1})$, implying $r= d_m(\tau)$; with $\Tr \sum_n \phi^0_n(\mathbb{1}) = D$ one finds $a=1$.  

We note that the first part of condition (Ji) implies that the operation $\Phi^0= \sum_n\phi^0_n$ characterizing the non-selective measurement by the first instrument $\{\phi^0_n \}$ is unital, i.e., $\Phi^0(\mathbb{1}) = \mathbb{1}$. This statement equivalently expresses the fact that the von Neumann entropy does not decrease under the non-selective measurement operation $\Phi^{0}$ \cite{Hayashi}.
Hence, this condition demands that the first measurement being performed non-selectively, does not decrease the von Neumann entropy.
This observation is compatible with the fact that the Jarzynski relation does not hold in its original form for feedback-control protocols~\cite{Sagawa,Morikuni,WVTCH}, where the von Neumann entropy may decrease.

The second part of condition (Ji) implies that the second measurement can be erroneous as long as $\Tr \phi^{\tau *}_m (\mathbb{1})=d_m(\tau)$, i.e., the trace of the effect of the second measurement is given by the degeneracy of each corresponding energy eigenvalue $e_m(\tau)$.
In fact, we have the following example of a quite erroneous energy measurement satisfying condition (Ji):
any instrument $\{\phi_m^{\tau}\}_{m\in\mathcal{N}_{\tau}}$ with the effect
\begin{align}
 \phi_m^{\tau*}(\mathbb{1}) = E_m^{Q} := \sum_i Q(m|i) \ketbra{\psi_i},
\end{align}
e.g.~$\phi_m^{\tau}(\rho) = \sqrt{E_m^{Q}}\rho \sqrt{E_m^{Q}}$,
where $\{\ket{\psi_i}(i=1,2,\cdots)\}$ is an arbitrary orthonormal basis and $Q(m|i)$ is a stochastic matrix, i.e., $\sum_m Q(m|i)=1$, satisfying $\sum_i Q(m|i) = d_m(\tau)$.
Especially, if $H(\lambda(\tau))$ is nondegenerate, $Q(m|i)$ is an arbitrary doubly stochastic matrix.
As $\{\ket{\psi_i}(i=1,2,\cdots)\}$ is an arbitrary basis and $Q(m|i)$ is generic enough, such a measurement can be quite erroneous.

Together with the requirement that the first measurement has to be error free, the condition (Jii) implies that the conditional probability
$p(m|n) =\Tr \phi_m^\tau\big (U_\Lambda \phi^0_n(\Pi_n) U^\dagger_\Lambda \big) = d_m(\tau)/D$ 
is independent of $n$. Therefore, the joint probability factorizes as $p(m,n) =(d_m(\tau)/D)\hspace{1pt} p^0_\beta(n)$ with $p^0_\beta(n) = Z^{-1}(0) \hspace{1pt}d_n(0)\hspace{1pt} e^{-\beta e_n(0)}$, and hence the results of the first and the second energy measurements are independent of each other. As a consequence, all possible correlations caused by the time-evolution throughout the force protocol would be suppressed by a measurement instrument that complies with the condition (Jii). Even though such an instrument satisfies the Jarzynski equality, it can hardly be considered an useful work meter.

For systems with a separable infinite-dimensional Hilbert space, (Jii) is excluded since $D$ is infinity and $\phi_m^{\tau*}(\mathbb{1}) = 0$ is not eligible as an effect.
Therefore, from Lemma \ref{Lem_Jarzynski}, we obtain the following necessary and sufficient condition for Eq.~(\ref{Key_JE_0}) to hold:
\begin{enumerate}
 \item[(Ji$^*$)] $\sum_n \phi_n^0(\mathbb{1}) = a \mathbb{1}$ and $\Tr \phi^{\tau *}_m(\mathbb{1}) ={d_m(\tau)}/{a}$, with a positive number $a$\:.
\end{enumerate}

It is straightforward to demonstrate that each of the conditions (Ji), (Jii), and (Ji$^*$) leads to the Jarzynski equality. Hence these constitute the sought-for necessary and sufficient conditions on any universal measurement scheme for the validity of the Jarzynski equality.
We note that even if the unitary dynamics $U_{\Lambda}\rho U_{\Lambda}^{\dagger}$ is replaced with a unital quantum channel $\mathcal{U}(\rho)$, the same conditions are sufficient for the Jarzynski equality. Thus, the same statement also holds for unital quantum channels as the dynamics.
This class of open-system dynamics is known to satisfy fluctuation relations within TPEMS \cite{Rastegin13, Smith, Rastegin18}.

In closing this section, we discuss the consequences of the simultaneous validity of the forward and backward Jarzynski equalities given by Eqs.~(\ref{JE}) and (\ref{bJE}), respectively.   
As the first consequence, we infer that the first measurement instruments of the forward and the backward processes must be error free. Because the measurement operations of these two processes are mutually time-reversed to each other as expressed in Eq.~(\ref{tr}), the second measurement of the forward process corresponds to the first measurement of the backward process, which is error free.
Consequently, the second measurement of the forward process must be error free and has to satisfy either (Ji) or (Jii) in the finite-dimensional case, or (Ji*) for systems with a separable infinite Hilbert space.
Hence, the operations for the first and the second measurements of the forward process are of the form
\begin{equation}
\phi^t_m(\rho) = \mathcal{E}^t_m\big (\Pi_m(t) \rho \Pi_m(t)\big ) \quad (t =0,\tau).
\label{phit}
\end{equation}   
Here $\{\mathcal{E}^t_m\}$ are trace preserving completely positive maps on $TC(\mathcal{H})$. 
Additionally, the measurement operations must satisfy either the condition (Ji) or (Jii) for systems with finite-dimensional Hilbert spaces. In the case of (Ji), the second part of the condition is automatically satisfied.
In the infinite-dimensional case,
because the second measurement is also error free,
$\phi^{\tau *}_m(\mathbb{1})=\Pi_m(\tau)$ holds as proved in Appendix \ref{errorfree}.
Hence, $\Tr \phi_m^{\tau *} (\mathbb{1}) =d_m(\tau)$ yields $a=1$.
  
  \section{Crooks relation} \label{Crooks}
  \subsection{Necessary and sufficient conditions}
As noted in the introduction, the Crooks relation implies the Jarzynski equalities for the forward and the backward processes given by Eqs. (\ref{JE}) and (\ref{bJE}), respectively. Therefore, the necessary conditions on the measurement operations for the validity of the Jarzynski equalities must equally hold for the Crooks relation implying that the first operations of both forward and the backward protocol must be error free.
As a consequence, the work pdfs for the forward and the backward protocols can be expressed as
\begin{align}
 p_\Xi(w) = \sum_{m,n} \delta\big (w -e^\Xi_m(\tau) + e^\Xi_n(0)\big ) & p_\Xi(m, n |n)  p^\Xi_n \nonumber\\
 &\hspace{10pt}(\Xi = \Lambda,\bar{\Lambda})\:,
\label{pXw}
\end{align}
where $e^\Lambda_m(t) =e_m(t)$, $e^{\bar{\Lambda}}_m(t) = e_m(\bar{t})$ with $t= 0,\tau$, and $\bar{0}=\tau$ as well as $\bar{\tau} =0$. Further, $p^\Lambda_n =Z^{-1}(0) \hspace{1pt}e^{-\beta e_n(0)} d_n(0)$ and $p^{\bar{\Lambda}}_n = Z^{-1}(\tau)\hspace{1pt} e^{-\beta e_n(\tau)} d_n(\tau)$. 
Putting these expressions into the Crooks relation~(\ref{CR}), we get
\begin{equation}
\begin{split}
&\sum_{m,n} \delta \big (w - e_m(\tau) + e_n(0) \big ) e^{-\beta e_n(0)}\\
& \times \left [ p_\Lambda(m,n|n) d_n(0) - p_{\bar{\Lambda}}(n, m|m) d_m(\tau) \right ] =0\:.
\end{split}
\label{pmn0}
\end{equation}  
According to the universality requirements (a) and (b) given in Section \ref{results}, Eq.~(\ref{pmn0}) must hold for all possible energy spectra of the initial and the final Hamiltonian. In particular, this must be satisfied for those spectra for which the allowed work values are not degenerate in the sense that $e_m(\tau) - e_n(0) = e_k(\tau) - e_l(0)$ holds only if $m=k$ and $n=l$. Hence, in order that Eq.~(\ref{pmn0}) is satisfied for all possible work values, the expression in the square bracket must vanish for all $m$ and $n$. Therefore, as a necessary condition for the Crooks relation to hold, the following generalized detailed balance condition~\cite{TMYH} must be satisfied:
\begin{equation}      
p_\Lambda(m,n|n) d_n(0) = p_{\bar{\Lambda}}(n, m|m) d_m(\tau)
\label{dbr}
\end{equation}
in accordance with the previous findings in \cite{VWT}.

Specializing the conditional probabilities as given by Eq.~(\ref{pLmnk}) for error-free measurements, we obtain
\begin{equation}
\begin{split}
p_\Lambda(m,n|n) &= \Tr \phi^\tau_m \big(U_\Lambda \phi^0_n(\Pi_n(0)) U^\dagger_\Lambda \big ) /d_n(0),\\
p_{\bar{\Lambda}}(m,n|n)&= \Tr \bar{\phi}^\tau_m \big (U_{\bar{\Lambda}} \bar{\phi}^0_n(\bar{\Pi}_n(\tau)) U^\dagger_{\bar{\Lambda}} \big ) /d_n(\tau)\\
&= \Tr \phi^0_m\big ( U^\dagger_\Lambda \phi^\tau_n \big (\Pi_n(\tau) \big ) U_\Lambda \big )/d_n(0)\:.
\end{split}
\label{pmn}
\end{equation}
In going to the last line, we have used the time reversal relations for the time-evolution operator, $\theta^\dagger U_\Lambda \theta = U^\dagger_{\bar{\Lambda}}$, for the projection operator $\bar{\Pi}_n(\tau) = \theta \Pi_n(\tau) \theta^\dagger$, and for the measurement operation (\ref{tr}). Using the expression (\ref{phit}) for the measurement operation together with (\ref{pmn}), we obtain the following condition on the trace preserving maps $\mathcal{E}^t_n$ from the detailed balance relation (\ref{dbr}):
\begin{equation}
\Tr U_\Lambda \mathcal{E}^0_n\big (\Pi_n(0) \big) U^\dagger_\Lambda \Pi_m(\tau) = \Tr U^\dagger_\Lambda \mathcal{E}^\tau_m\big (\Pi_m(\tau) \big ) U_\Lambda \Pi_n(0)\:.
\label{db2}
\end{equation}
According to the universality requirement (c) formulated in Section~\ref{results}, this condition must be satisfied for all possible protocols connecting the respective initial and final Hamiltonians. In particular, the protocol may consist of an initial sudden quench from $H\big (\lambda(0)\big )$ to an arbitrary Hamiltonian $H$ which may be kept constant until the final time when the Hamiltonian is suddenly quenched to the final Hamiltonian $H\big (\lambda(\tau) \big )$. Because the sudden quenches at the beginning and the end of the force protocol do not contribute to the time evolution operator, for this particular protocol it is given by $U_\Lambda =e^{-(i \tau/\hbar) H}$. With a proper choice of $H$, any unitary operator $U$ can be realized in this way. Hence, the detailed balance-like condition (\ref{db2}) must be satisfied for any unitary operator $U_\Lambda$.       
Conversely, Eq.~(\ref{db2}) presents an equivalent formulation of the detailed balance relation Eq.~(\ref{dbr}), which implies the validity of Eq.~(\ref{pmn0}) and consequently that of the Crooks relation~(\ref{CR}) as well. Therefore, the Crooks relation universally holds if and only if the measurement operations $\phi^t_m$ with $t=0$ and $\tau$ are error free and satisfy Eq.~(\ref{db2}) for any protocol $\Lambda$.

  \subsection{Nondegenerate case}\label{nondegenerate}
In discussing the condition (\ref{db2}), we first restrict ourselves to initial and final Hamiltonians with non-degenerate spectra. Hence, the projection operators $\Pi_n(t) =|n;t\rangle \langle n;t|$ with $t=0$ and $\tau$ yields pure states. The condition (\ref{db2}) then becomes
\begin{equation}
\begin{split}
&\langle m;\tau|U_\Lambda \mathcal{E}^0_n\big (|n;0 \rangle \langle n;0| \big) U^\dagger_\Lambda |m; \tau \rangle  \\
=& \langle n;0|U_\Lambda^{\dagger} \mathcal{E}^\tau_m \big (|m;\tau \rangle \langle m;\tau| \big ) U_\Lambda |n; 0 \rangle.
\end{split}
\label{db3}
\end{equation}
Because it must be satisfied for all unitary operators $U_\Lambda$, one finds, using Lemma 4 from Appendix \ref{PoKL}, that Eq.~(\ref{db3}) holds if and only if the trace preserving, completely positive operations $\mathcal{E}^t_n$ act on the associated projection operators as
\begin{equation}
\mathcal{E}^t_n\big (|n;t \rangle \langle n;t| \big ) = \frac{1-\alpha}{D} \mathbb{1} + \alpha   |n;t \rangle \langle n;t|\:.
\label{Etn}
\end{equation}
Here, we additionally have that the Hilbert space is of finite dimension $D= \dim (\mathcal{H})$ and the parameter $\alpha$ satisfies $-(D-1)^{-1} \leq \alpha \leq 1$. In combination with the expression $\phi_n^{t}(\rho) = \mathcal{E}_n^{t}(\ketbra{n;t}\rho\ketbra{n;t}) = \bra{n;t}\rho\ket{n;t}\mathcal{E}_n^{t}(\ketbra{n;t})$, we obtain
\begin{equation}
\phi^t_n(\rho) =\langle n;t|\rho| n;t \rangle  \left ( \frac{1-\alpha}{D}  \mathbb{1} + \alpha|n;t\rangle \langle n;t| 
\right ).
\label{phitn}
\end{equation}
For systems with a separable infinite-dimensional Hilbert space, the only universal measurements are projective and hence given by
\begin{equation}
\phi^t_n(\rho) = \langle n;t |\rho | n;t \rangle |n;t\rangle \langle n;t|.
\label{phip}
\end{equation}

The action \eqref{phitn} can be accomplished by a so-called depolarizing channel $\mathcal{E}_\alpha$~\cite{Hayashi, channel} acting as
\begin{equation}
\mathcal{E}_\alpha(\rho) =  (1-\alpha)\frac{\Tr \rho}{D} \mathbb{1} + \alpha \rho  \:
\label{El}
\end{equation}
with $-(D^2 -1)^{-1} \leq \alpha \leq 1$ or, likewise by a transpose depolarizing channel~\cite{DHS}
\begin{equation}
\Delta_\alpha(\rho) =  (1-\alpha)\frac{\Tr \rho}{D} \mathbb{1} + \alpha \rho^T  \:,
\label{Elt}
\end{equation}
where the transpose operation $^T$ is defined with respect to the energy eigenbasis and $-(D-1)^{-1}\leq \alpha \leq (D+1)^{-1}$. We note that the different ranges of $\alpha$ guarantee that the resulting operations are completely positive.

In fact, note that it does not matter how the channels $\mathcal{E}^t_n$ are globally defined as long as they act on the energy eigenstate $|n;t \rangle \langle n;t|$ according to Eq.~(\ref{Etn}). The choice made in Eqs. (\ref{El}) and (\ref{Elt}) has a remarkable property of being independent of the index $n$ specifying the argument of the channel. While the transpose depolarizing channel $\Delta_\alpha$ still depends on the basis with respect to which the transpose operation is performed, the depolarizing channel $\mathcal{E}_\alpha$ is completely independent of the basis and hence can be considered as a  {\it universal channel}.
We note that $\mathcal{E}_{\alpha}$ is the only universal channel compatible with Eq.~(\ref{Etn}). For a universal channel, Eq.~(\ref{Etn}) must also hold if $|n;t\rangle$ is replaced by an arbitrary normalized element of the Hilbert space. By linearity, then Eq.~(\ref{El}) follows.
Thus, with the universal channel~(\ref{El}), the range of $\alpha$ is restricted to
$-(D^2-1)^{-1} \leq \alpha \leq 1$.

  \subsection{With degeneracy}\label{w_degeneracy}
  Even if the initial or the final Hamiltonian has degenerate spectra, one finds by inspection that the TEMS with the measurement operations defined as
  \begin{equation}
\phi^t_n(\rho) = \left \{\begin{array}{ll}
&(1-\alpha)D^{-1}\Tr [\Pi_n(t) \rho] \mathbb{1}  + \alpha \Pi_n(t) \rho \Pi_n(t)\\*[1mm]
 &\hspace{30pt}(-(D^2 -1)^{-1} \leq \alpha \leq 1 \text{ and }  D < \infty),\vspace{5pt}\\*[2mm]
 &\Pi_n(t) \rho \Pi_n(t) \quad (D= \infty)
\end{array}
\right .
\label{phitnu}
\end{equation}
 implies the validity of the Crooks relation.
 We note that the parameter $\alpha$ must have the same value for the initial and the final measurements.

 On the other hand, we obtain Eq.~\eqref{phitnu} as a necessary condition for the Crooks relation to hold regardless of the degeneracy of the initial and final Hamiltonians if we restrict ourselves to the requirements (a)--(c) and additionally to universal channels, satisfying $\mathcal{E}^t_n =\mathcal{E}^{t}$ independently of the eigenprojections of each Hamiltonian.
 That is verified by observing that $\mathcal{E}_n^t$ should be the same as that for nondegenerate Hamiltonians in order that the channels are independent of the eigenprojections of each Hamiltonian.
  Hence, the argument for nondegenerate Hamiltonians in the previous section uniquely specifies $\mathcal{E}_n^{t}$ to be the depolarizing channel $\mathcal{E}_{\alpha}$ with $-(D^2 - 1)^{-1} \leq \alpha \leq 1$.
  Therefore, Eq.~(\ref{phitnu}) represents necessary and sufficient conditions for the restricted class of universal measurement operations with universal channels.
 
We note that in the finite-dimensional case given by the first line of Eq.~(\ref{phitnu}), projective measurements contribute with a factor $\alpha$, which may even be negative, while the first term with the factor $(1-\alpha)$ leaves a maximally mixed state after the measurement. Furthermore, for the conditional probability as specified in Eq.~(\ref{pmn}), one obtains a mixture of the projective two-point measurement result $p^{2p}_\Lambda(m|n) := \Tr \Pi_m(\tau) U_\Lambda \Pi_n(0) U^\dagger_\Lambda d^{-1}_n(0)$ and a totally random choice of the second energy. Hence, one finds 
\begin{equation}
p_\Lambda(m,n|n) = (1-\alpha)\frac{d_m(\tau)}{D} + \alpha p^{2p}_\Lambda(m|n)\:, 
\label{pmnm}
\end{equation}     
with $-(D^2-1)^{-1} \leq \alpha \leq 1$. While the contribution from the projective measurements depends on the dynamics, the other part is completely independent of any dynamics taking place during the force protocol, therefore giving energy values by mere chance.  Finally, in the case of separable infinite-dimensional Hilbert space, only projective energy measurements are allowed.

  \section{Conclusion}\label{Conclusion}
We have discussed necessary and sufficient conditions that have to be imposed on generalized energy measurements determining the work exerted on a system by a force protocol in order that the fluctuation theorems of Jarzynski and Crooks are obeyed like in the case of projective energy measurements. In our analysis, we restricted ourselves to a class of universal measurements that are characterized by measurement operations being independent of the spectra of the Hamiltonians of the system immediately before and after the force protocol, as well as of the temperature of the initial state.

When restricting to these universal measurements, a necessary condition for the Jarzynski equality to hold requires that the first measurement is error free.
The backaction of any error-free measurement can be visualized as a projective measurement followed by a quantum channel.
Remarkably, we find that the Jarzynski equality is satisfied even if the second measurement is quite erroneous as long as the associated effects satisfy a normalization condition [see J(i) in Sec.~\ref{Jarzynski}].
Although a quantum work based on such an erroneous measurement does not coincide with the energy difference that would be obtained with projective energy measurements,
it may still be useful for the estimation of the free energy difference. That is a possible advantage of a generalized notion of work compatible with the Jarzynski equality.
If one additionally requires that the Jarzynski equality is fulfilled for the backward process, then the second energy measurement has to also be error free.
This fact yields that both measurements must be error free for the Crooks relation
since the Crooks relation implies the validity of the Jarzynski equality for both processes.

The Crooks relation requires a generalized detailed balance relation. It connects the forward and backward conditional probabilities based on generalized energy measurements.
For processes with non-degenerate initial and final Hamiltonians, this detailed balance relation entails that the quantum channel following the projective measurement must be a depolarizing channel if it is independent of the protocol.
Especially, this channel turns out to be independent of the outcome.
For systems with a separable infinite-dimensional Hilbert space, there are no non-trivial depolarizing channels and hence both energy measurements have to be projective.

For processes which connect degenerate Hamiltonians, we were not able to extract a necessary condition on the channels following the projective measurements within the universality requirements (a) -- (c) of Section \ref{results}.
From a practical point of view, it would be quite inconvenient if different measurement apparatuses had to be employed depending on whether the energy spectrum contains degeneracies, a detail that may not be known a priori. Therefore, the restriction to universal channels seems quite natural. With such a restriction, we have extended the condition on the Crooks relation derived for systems with non-degenerate energy spectra to the degenerate case. Nevertheless, it would be interesting to explore if non-universal channels could also lead to measurement operations that are distinct from the form (\ref{phitnu}) but still allow the Crooks relation to be satisfied. To reach this goal, we would need an analog of Lemma 4 to deal with Hamiltonians that have a degenerate spectrum. This presents a technically difficult open problem.

It is interesting to note that the compatibility of a quantum work with fluctuation theorems requires a backaction of the energy measurements which makes it impossible to satisfy condition (i) of the Perarnau-Llobet {\it et al.} no-go theorem \cite{Perarnau-Llobet}: Any coherence with respect to the energy basis of a non-equilibrium initial state is erased by the first error-free measurement. That is, if we apply a two-energy measurement scheme compatible with fluctuation theorems to a non-thermal initial state,
the obtained average work differs from the difference in the average energy between the final and the initial states in general.

Finally,
the requirement of error-free measurements appearing in the conditions for the validity of fluctuation theorems might
put a severe restriction on the experimental realization and the practical usefulness of the fluctuation relations for quantum mechanical systems.
However, the necessity of our conditions for the validity of fluctuation theorems
suggests a potential application of fluctuation theorems to a diagnostic verification of the accuracy of quantum measurements in a similar spirit to \cite{Gardas}.
Moreover, other measurement strategies that can yield erroneous results may still be of practical use as long as the errors are superimposed on the sought-for work distribution in an a priori known way that makes a correction possible after the measurement. This strategy can be employed in particular for Gaussian energy measurements \cite{WVT} and Gaussian work meters~\cite{TH2016,Cerisola}.

 \begin{acknowledgements}
 We thank Akira Shimizu for helpful discussions and comments in the early stage of this work.
 We were supported by NSF of China (Grant No. 11674283), by the Zhejiang Provincial Natural Science Foundation Key Project (Grant No. LZ19A050001), by the Fundamental Research Funds for the Central Universities (Grants No. 2017QNA3005 and No. 2018QNA3004), by the Zhejiang University 100 Plan, and by the Thousand Young Talents Program of China. B.P.V. is supported by the Research Initiation Grant and Excellence-in-Research Fellowship of IIT Gandhinagar.
 \end{acknowledgements}

  \appendix
 \section{Error-free instruments}\label{errorfree}
From the condition $p(m|n) =\delta_{m,n}$ for an instrument to be error-free it follows with Eq.~(\ref{pmn}) that $\Tr \phi_m(\Pi_n) = \Tr \Pi_n \phi^*_m(\mathbb{1}) =0$ for all $n\neq m$, or, equivalently $\langle \psi |\phi^*_m(\mathbb{1})| \psi \rangle =0$ for all $\psi \in (\mathbb{1} - \Pi_m) \mathcal{H}$. Defining $|\varphi\rangle = |\psi \rangle + z |v \rangle$ with an arbitrary $|v \rangle \in \mathcal{H}$ and arbitrary $z \in \mathbb{C}$, we obtain the following inequality from the positivity of $\phi^*_m(\mathbb{1})$:
 \begin{equation}
\begin{split}
\langle \varphi|\phi^*_m(\mathbb{1})|\varphi \rangle & = \langle v |\phi^*_m(\mathbb{1})|v \rangle |z|^2 \\
&\quad + \langle \psi |\phi^*_m(\mathbb{1})|v \rangle z  + \langle v |\phi^*_m(\mathbb{1})|\psi \rangle \bar{z}\\
&\geq 0
\end{split}
 \label{cerf}
\end{equation} 
where $\bar{z}$ denotes the complex conjugate of $z$. Because this inequality must hold for all $z\in \mathbb{C}$, $\langle v, |\phi^*_m(\mathbb{1})|\psi\rangle =0$ must hold for all $|v \rangle \in \mathcal{H}$ and all $|\psi \rangle \in (\mathbb{1} - \Pi_m) \mathcal{H}$ yielding $\phi^*_m(\mathbb{1}) (\mathbb{1} - \Pi_m) =0$. Multiplying from the right by $\Pi_n$ one obtains $\phi^*_m(\mathbb{1}) \Pi_n = \delta_{m,n} \phi^*_n(\mathbb{1})$ and thus $\phi^*_m(\mathbb{1}) = \Pi_m$ follows in combination with the completeness relation $\sum_{m}\phi^{*}_m(\mathbb{1}) = \mathbb{1}$. Multiplying both sides with an arbitrary $\rho \in TC(\mathcal{H})$  and taking the trace, one obtains $\Tr \phi_m(\rho) = \Tr \Pi_m \rho$.
Consequently the operation $\phi_m$ consists of the subsequent application of the projective measurement and a completely positive trace preserving operation $\mathcal{E}_m$ \cite{Hayashi} in accordance with Eq.~\eqref{ef}.
  
\section{Key lemmas}\label{PoKL}
Here we present the technical background which is needed to derive the necessary conditions for the  Jarzynski equality and the Crooks relation. 
We begin with three lemmas on which the necessary conditions (JI) and (JII) for the Jarzynski equality are based.

In the following, we consider separable Hilbert spaces, which include both finite and infinite-dimensional spaces.
Both cases are treated in parallel as far as the dimensionality does not matter. In the sequel, only bounded operators are considered presenting no restriction for finite dimensional Hilbert spaces.

\subsection{Jarzynski equality}    

   \begin{lemma}\label{Lem_nee}
    Let $A$ and $B$ be bounded operators on a separable Hilbert space $\mathcal{H}$.
   Suppose that $A$ is not proportional to the identity $\mathbb{1}$.
   Then, any vector $\ket{v}\in\mathcal{H}$ which is not an eigenvector    
   of $A$ is an eigenvector of $B$ if and only if $B$ is proportional to  
   $\mathbb{1}$.
  \end{lemma}
   \begin{proof}
    We prove the ``only if'' part since the ``if'' part is trivial.
    If $\mathcal{H}$ is infinite-dimensional, the bounded operator $A$ may have no eigenvalues in the sense that any $\lambda$ satisfies $\mathrm{Ker} (\lambda \mathbb{1} - A) = \{0\}$ even if $\lambda \mathbb{1} - A$ is not invertible; i.e., $\lambda$ belongs to the spectrum of the operator $A$ that has an empty point-spectrum \cite{Yosida}.
    In this case, $A$ does not possess any eigenvector $| v \rangle \in \mathcal{H}$ and hence, any vector is an eigenvector of $B$, according to the assumption.
   Thus, $B$ is proportional to $\mathbb{1}$.
   
Now suppose that $A$ has an eigenvector $\ket{a}$ with corresponding  
    eigenvalue $a$.
    In the following, we parallelly deal with both finite and infinite-dimensional Hilbert spaces since the proof is similar regardless of its dimensionality.
    
   Since $A$ is not proportional to $\mathbb{1}$, there exists a vector  
    $\ket{v}\neq 0$ which is not an eigenvector of $A$.
    Let us construct a sequence $\epsilon_n \rightarrow 0$ such that $\ket{a} + \epsilon_n \ket{v}$ is not an eigenvector of $A$.
    To do so, we show that there is at most a single value of $\epsilon$  that can make $\ket{a} + \epsilon \ket{v}$ an eigenvector of $A$.
    Suppose that there exists a number $\epsilon > 0$ such that
    $\ket{a}+\epsilon\ket{v}$ is an eigenvector of $A$ with its eigenvalue $\lambda$.
    Then, $A(\ket{a}+\epsilon\ket{v}) = \lambda (\ket{a}+\epsilon\ket{v})$ and $A\ket{a} = a \ket{a}$ read
    $
     (a-\lambda)\ket{a} + \epsilon (A -\lambda\mathbb{1})\ket{v} = 0,
    $
    which implies
    \begin{align}
     (A-\lambda\mathbb{1})\ket{v} = c_0 \ket{a} \label{LDep2}
    \end{align}
    with
    \begin{align}
     c_0 = \frac{(\lambda - a)}{\epsilon}.\label{c_0Def}
    \end{align}
    Equation \eqref{LDep2} further reads
    \begin{align}
     A\ket{v} = \lambda \ket{v} + c_0 \ket{a}.\label{LvL0a}
    \end{align}
    Hence, $A\ket{v}$ is in the subspace spanned by two linearly independent vectors $\ket{v}$ and $\ket{a}$.
    Since such a decomposition $A\ket{v} = c_v \ket{v} + c_a \ket{a}$ is unique,
    $\lambda$ and $c_0$ are uniquely identified as $\lambda = c_v$ and $c_0 = c_a$.
    %
    Noting that $A\ket{v}$, $\ket{v}$, and $\ket{a}$ are all independent of $\epsilon$, we can conclude that the eigenvalue $\lambda$ with the eigenvector of the form $\ket{a}+\epsilon\ket{v}$ is unique.
    Therefore, from equation \eqref{c_0Def} and the fact that $c_0=c_a$ is unique, at most only one number $\epsilon = (c_v - a)/c_a$ can make $\ket{a} + \epsilon \ket{v}$ an eigenvector of $A$.
    %
    Thus, taking a sequence $\epsilon_n$ $(n\geq 1)$:
    \begin{align}
     \epsilon_n :=
     \left \{
     \begin{array}{ll}
      \frac{\epsilon}{n+1} \quad (\exists \epsilon>0, \ket{a} + \epsilon \ket{v} \text{ is an eigenvector of } A) \\
      \frac{1}{n} \quad \mathrm{(otherwise)},
     \end{array}
     \right.
    \end{align}    
we have a sequence $\ket{v_n}\rightarrow \ket{a}$ of vectors:
   \begin{align}
    \ket{v_n}:= \ket{a} + \epsilon_n \ket{v},
   \end{align}
   each of which is not an eigenvector of $A$.
   Because $\ket{v_n}$ is not an eigenvector of $A$, the relation $B\ket{v_n} = b_n\ket{v_n}$ holds with some number $b_n$ by the assumption.
    We have
    \begin{align}
     B\ket{a} = \lim_{n\rightarrow\infty} B\ket{v_n}
     = \lim_{n\rightarrow\infty} b_n\ket{v_n}.
    \end{align}
    Because $(b_n\ket{v_n})$ converges to a vector $\ket{w}=B\ket{a}$, $(b_n)$ is bounded, say $|b_n| \leq b$.
    Thus, for any vector $\ket{u}\perp \ket{a}$, we have
    \begin{align}
     |\braket{u}{w}| =& \lim_{n\rightarrow\infty}|b_n\bra{u}\ket{v_n}|\nonumber\\
     \leq& b \lim_{n\rightarrow\infty}|\braket{u}{v_n}|
     = |\braket{u}{a}| = 0.
    \end{align}
    Thus, $\ket{w}=c\ket{a}$ with some number $c$, which means that
    \begin{align}
     B\ket{a} = \ket{w} = c\ket{a}.
    \end{align}
    Therefore, any vector of the Hilbert space $\mathcal{H}$ is an eigenvector of $B$, which is only possible if $B$ is proportional to $\mathbb{1}$.
   \end{proof}
  
   \begin{lemma}\label{Lem_zero}
   Let $A$ and $B$ be bounded operators on a separable Hilbert space $\mathcal{H}$.
   If neither $A$ nor $B$ is proportional to the identity $\mathbb{1}$, then there exists an orthonormal system $\{\ket{e_1},\ket{e_2}\}\subset \mathcal{H}$ such that
   $\bra{e_2}A\ket{e_1}\neq 0$ and $\bra{e_2}B\ket{e_1}\neq 0$ hold.
  \end{lemma}
  \begin{proof}
   Because neither $A$ nor $B$ is proportional to $\mathbb{1}$, there exists a unit vector $\ket{e_1}$ which is not an eigenvector of $A$ nor of $B$ from Lemma \ref{Lem_nee}.
   Thus, the subspace $\mathcal{V}:=\mathrm{span}\{\ket{e_1},A\ket{e_1},B\ket{e_1}\}$ is not one dimensional.
   We take a unit vector $\ket{\tilde{e}_2}\in\mathcal{V}$ which is orthogonal to $\ket{e_1}$.
   Let us show that we can always construct the desired vector $\ket{e_2}$.
   If $\bra{\tilde{e}_2}A\ket{e_1}\neq 0$ and $\bra{\tilde{e}_2}B\ket{e_1}\neq 0$, $\{\ket{e_1},\ket{\tilde{e}_2}\}$ is the desired orthonormal system.
   If this is not the case, we have $\bra{\tilde{e}_2}A\ket{e_1} = 0$ without loss of generality.
   In this case, since $A\ket{e_1}$ is not proportional to $\ket{e_1}$, there exists another unit vector $\ket{e_3}\perp \ket{e_1}$ which is composed of a linear combination of $\ket{e_1}$ and $A\ket{e_1}$.
   Then, $\{\ket{e_1},\ket{\tilde{e}_2},\ket{e_3}\}$ is a completely orthonormal system (CONS) of $\mathcal{V}$.
   We define a projection $\Pi := \ketbra{\tilde{e}_2} + \ketbra{e_3}$.
   Note that $\Pi B\ket{e_1}\neq 0$ since $B\ket{e_1}$ is not proportional to $\ket{e_1}$.
   Then, taking a real number $x$ which satisfies
   \begin{align}
    x\neq -\frac{\bra{e_3}B\ket{e_1}}{\|\Pi B\ket{e_1}\|^2},
    -\frac{1}{\bra{e_1}B\ket{e_3}},
   \end{align}
   we define a unit vector as follows:
   \begin{align}
    \ket{e_2} := c (\ket{e_3} + x \Pi B \ket{e_1}),
   \end{align}
   where $c$ is the normalization factor.
   Because $\ket{e_2}\in \Pi \mathcal{V}$, we have $\ket{e_2}\perp\ket{e_1}$.
   We have
   \begin{align}
    \bra{e_2}A\ket{e_1}
    = c \bra{e_3}A\ket{e_1}(1 + x\bra{e_1}B\ket{e_3}) \neq 0
   \end{align}
   and
   \begin{align}
    \bra{e_2}B\ket{e_1}
    = c (\bra{e_3}B\ket{e_1} + x\|\Pi B \ket{e_1}\|^2) \neq 0
   \end{align}
   by the definition of $x$.
   Therefore $\{\ket{e_1},\ket{e_2}\}$ is a desired orthonormal system.
  \end{proof}
  
  \begin{lemma}\label{Lem_Jarzynski}
   Let $A$ and $B$ be \sout{a} self-adjoint bounded operators, and $U$ an arbitrary unitary operator on a separable Hilbert space $\mathcal{H}$.
   Then, the relation  
   \begin{align}
    \tr U^{\dagger}AU B = r \quad (\forall U\mathrm{: unitary})\label{LJ_0}
   \end{align}
   holds with an $U$-independent number $r\neq 0$, if and only if
   $A=a \mathbb{1}$ and $\tr B = r/a$ with a real number $a\neq 0$ (or $B=a \mathbb{1}$ and $\tr A = r/a$ by the symmetry between $A$ and $B$).
   For $r=0$, equation (\ref{LJ_0}) holds if and only if either the above holds, or $A$ or $B$ equals to $0$.
  \end{lemma}
  \begin{proof}
   The ``if'' part is obvious. We prove the ``only if'' part, i.e., we suppose that \eqref{LJ_0} holds.
   For an arbitrary CONS $\{\ket{f_k}|\}$ of $\mathcal{H}$, we take a unitary operator $U_x$ depending on a real number $x$ as follows:
   \begin{align}
    U_x := \sum_{k} e^{i x \phi_k} \ketbra{f_k},
   \end{align}
   where $\phi_k$ $(k=1,2,\cdots)$ are real numbers such that
   $\phi_{k_1} - \phi_{l_1} \neq \phi_{k_2} - \phi_{l_2}$ holds for any $k_1, k_2, l_1, l_2$ with $k_1 \neq l_1$, $(k_1,l_1)\neq (k_2,l_2)$.
   Then, the condition \eqref{LJ_0} reads
   \begin{align}
    \tr U_x^{\dagger} A U_x B
    = \sum_{k,l} e^{-ix (\phi_k - \phi_l)}\bra{f_k}A\ket{f_l}\bra{f_l}B\ket{f_k}
    = r.\label{lin_ind}
   \end{align}
   Since the functions $e^{-ix(\phi_k- \phi_l)}$ of $x$ and the identity function are linearly independent because of the definition of $\phi_k$,
   equation \eqref{lin_ind} holds for any $x$ if and only if
   \begin{align}
    \bra{f_k}A\ket{f_l}\bra{f_l}B\ket{f_k} =& 0 \quad (k\neq l)\label{klkl}\\
    \sum_{k}\bra{f_k}A\ket{f_k}\bra{f_k}B\ket{f_k} - r =& 0
   \end{align}
   hold.
   By applying Lemma \ref{Lem_zero}, $A$ or $B$ must be proportional to the identity $\mathbb{1}$ because equation \eqref{klkl} holds for any CONS $\{\ket{f_k}\}$ of $\mathcal{H}$.
   Otherwise, there exists an orthonormal system $\{\ket{f_1},\ket{f_2}\}$ such that $\bra{f_1}A\ket{f_2}\bra{f_2}B\ket{f_1}\neq 0$ from Lemma \ref{Lem_zero}.
   Thus, the proof is completed.
  \end{proof}

\subsection{Crooks relation}
  Next, we prove the key lemma for the necessary and sufficient condition for the Crooks fluctuation relation.
  
  \begin{lemma}\label{Lem_Crooks}
   Let $\rho$ and $\sigma$ be density matrices on a separable Hilbert space $\mathcal{H}$. Let $\ket{a}$ and $\ket{b}$ be pure states in this Hilbert space.
   For $D$-dimensional $\mathcal{H}$,
   \begin{align}
    \bra{a}U^{\dagger}\rho U\ket{a} = \bra{b}U\sigma U^{\dagger}\ket{b} \label{nd1}
   \end{align}
   holds for any unitary matrix $U$ if and only if there exists $\alpha$ with $-\frac{1}{D-1}\leq \alpha \leq 1$ such that
   \begin{align}
    (\rho,\sigma)
    = \left((1-\alpha)\frac{1}{D}\mathbb{1} + \alpha \ketbra{b}, (1-\alpha)\frac{1}{D}\mathbb{1} + \alpha \ketbra{a}\right),\label{rs}
   \end{align}
   where $\mathbb{1}$ is the identity.

For a separable infinite-dimensional Hilbert space $\mathcal{H}$
equation \eqref{nd1} holds for any unitary matrix $U$ if and only if $\rho = \ketbra{b}$ and $\sigma = \ketbra{a}$.
  \end{lemma}
  
  \begin{proof}
   To begin with, we prove the ``if'' part.
   The ``if'' part for the infinite-dimensional case is obvious.
   We focus on $D$-dimensional $\mathcal{H}$.
   In fact, for $-\frac{1}{D-1}\leq \alpha \leq 1$, $(\rho,\sigma) = \left((1-\alpha)\frac{1}{D}\mathbb{1} + \alpha \ketbra{b}, (1-\alpha)\frac{1}{D}\mathbb{1} + \alpha \ketbra{a}\right)$ are states which satisfy
   \begin{align}
    &\bra{a}U^{\dagger}\rho U\ket{a}\nonumber\\
    =&\bra{a}U^{\dagger}\left((1-\alpha)\frac{1}{D}\mathbb{1}+\alpha\ketbra{b}\right) U\ket{a}\nonumber\\
    =&(1-\alpha)\frac{1}{D}\braket{a} + \alpha\bra{a}U^{\dagger}\ket{b}\bra{b}U\ket{a} \nonumber\\
    =& (1-\alpha)\frac{1}{D} + \alpha\bra{a}U^{\dagger}\ket{b}\bra{b}U\ket{a}\nonumber\\
    =&(1-\alpha)\frac{1}{D}\braket{b} + \alpha\bra{b}U\ket{a} \bra{a}U^{\dagger}\ket{b}\nonumber\\
    =&\bra{b}U\left((1-\alpha)\frac{1}{D}\mathbb{1} + \alpha\ketbra{a}\right) U^{\dagger}\ket{b}\nonumber\\
    =&\bra{b}U\sigma U^{\dagger}\ket{b}.
   \end{align}
   Hence, ``if'' part is completed.

   Now, we prove ``only if'' part. Hence, suppose that a pair of states $(\rho,\sigma)$ satisfies \eqref{nd1}.
   At first, we consider the case where $\sigma = r_1 \mathbb{1} + r_2 \ketbra{a}$ with real numbers $r_1$ and $r_2$.
   In this case, condition \eqref{nd1} reads
   \begin{align}
    \bra{a}U^{\dagger}\rho U\ket{a} =  \bra{a}U^{\dagger}(r_1\mathbb{1} + r_2 \ketbra{b}) U \ket{a}
    \; (\forall U : \mathrm{unitary}).\label{nd1_p}
   \end{align}
   Since for any pure state $\ket{\psi}\in \mathcal{H}$ there exists a unitary operator $U$ such that $\ket{\psi}=U\ket{a}$,
   condition \eqref{nd1_p} leads to
   \begin{align}
    \bra{\psi}(\rho - r_1\mathbb{1} - r_2 \ketbra{b})\ket{\psi} = 0 \quad (\forall \ket{\psi}\in\mathcal{H}),
   \end{align}
   which implies that $\rho = r_1 \mathbb{1} + r_2 \ketbra{b}$.
   Similarly, $\rho = r_1 \mathbb{1} + r_2 \ketbra{b}$ also implies $\sigma = r_1 \mathbb{1} + r_2 \ketbra{a}$.
   Thus, these cases are covered by (\ref{rs}).
   
   Then, we suppose that $\sigma \neq  r_1 \mathbb{1} + r_2 \ketbra{a}$ for any $r_1, r_2$, and hence $\rho \neq  r_1 \mathbb{1} + r_2 \ketbra{b}$.
   When the pair of states $(\rho,\sigma)$ satisfies equation \eqref{nd1},
   if we replace $\rho$ and $\sigma$ with
   $\gamma\mathbb{1} - \rho$ and $\gamma\mathbb{1} - \sigma$ respectively
   for an arbitrary real number $\gamma \neq 0$, equation \eqref{nd1} still holds.   
   Because of the assumption that $\sigma \neq  r_1 \mathbb{1} + r_2 \ketbra{a}$, we also have
   $\gamma\mathbb{1} - \rho \neq  r_1' \mathbb{1} + r_2' \ketbra{b}$ and $\gamma\mathbb{1} - \sigma \neq  r_1' \mathbb{1} + r_2' \ketbra{a}$ for any $\gamma$, $r_1'$, and $r_2'$.

   Regardless of the finiteness of the dimensionality of $\mathcal{H}$, we have the spectral decomposition
   \begin{align}
    \rho &= \sum_{i} p_i \ketbra{v_i}\\
    \sigma &= \sum_{i} q_i \ketbra{u_i},
   \end{align}
   where their spectra $\{p_i\}$ and $\{q_i\}$ are displayed in descending order as $p_1\geq p_2 \geq \cdots$ and $q_1\geq q_2 \geq \cdots$.
   We assume that $p_1 \geq q_1$ without loss of generality.
   Then, taking $\gamma = p_1$, we have
   \begin{align}
    (p_1 \mathbb{1} - \rho)\ket{v_1} = 0.\label{v1egn}
   \end{align}
   From the assumption, we have $p_1 \mathbb{1} - \sigma \neq r \ketbra{a}$ for any real number $r$. Thus, there exists $i_0$ with $p_1 - q_{i_0} \neq 0$ such that $\ket{a}$ and $\ket{u_{i_0}}$ are linearly independent.
   Hence, $\mathcal{V}_0:=\mathrm{span}\{\ket{a},\ket{u_{i_0}}\}$ has a CONS $\{\ket{a}, \ket{g_1}\}$.
    We take a unitary $U_1=\ket{v_1}\bra{a} + \sum_{j}\ket{h_j}\bra{g_j}$, where $\{\ket{g_j} |j\geq 2\}$ is a CONS of $\mathcal{V}_0^{\perp}$ for $D>2$, and $\{\ket{h_j}\}$ is an arbitrary CONS of $\{\ket{v_1}\}^{\perp}$,
   where $\mathcal{V}^{\perp}:=\{\ket{u}\in\mathcal{H}| \braket{u}{v} = 0 \; (\forall \ket{v}\in\mathcal{V})\}$.
   Since $U_1$ satisfies $U_1\ket{a}=\ket{v_1}$ by definition, equation (\ref{nd1}) reads
   \begin{align}
    \bra{b}U_1(p_1\mathbb{1} - \sigma)U_1^{\dagger}\ket{b}
    = \bra{v_1}(p_1\mathbb{1} - \rho)\ket{v_1} = 0
   \end{align}
   because of \eqref{v1egn}.
   Therefore, we have
   \begin{align}
    \sum_{i} (p_1 - q_i) |\bra{b}U_1\ket{u_i}|^2 = 0.
   \end{align}
   Thus, $\bra{b}U_1\ket{u_i}=0$ holds for any $i$ with $p_1 - q_i \neq 0$ because $p_1 - q_i \geq 0$ following from $p_1\geq q_1$.
   Hence, we have
   \begin{align}
    0 = \bra{b}U_1\ket{u_{i_0}}
    = \bra{b}\ket{h_1}\bra{g_1}\ket{u_{i_0}} + \bra{b}\ket{v_1}\bra{a}\ket{u_{i_0}}
   \end{align}
   since $\bra{g_j}\ket{u_{i_0}}=0$ $(j\geq 2)$ follows from the definition.
   Because $\ket{a}$ and $\ket{u_{i_0}}$ are linearly independent, we have $\bra{g_1}\ket{u_{i_0}}\neq 0$.
   Therefore, we obtain
   \begin{align}
    \bra{b}\ket{h_1}=-\frac{\bra{b}\ket{v_1}\bra{a}\ket{u_{i_0}}}{\bra{g_1}\ket{u_{i_0}}}.\label{l1-5}
   \end{align}
   Since the right-hand side of \eqref{l1-5} is independent of $\ket{h_1}$, and $\ket{h_1}$ is an arbitrary vector in $\{\ket{v_1}\}^{\perp}$,
   we have $\bra{b}\ket{h_1} = -\bra{b}\ket{v_1}\bra{a}\ket{u_{i_0}}/\bra{g_1}\ket{u_{i_0}} = -\bra{b}\ket{h_1}$ by replacing $\ket{h_1}$ with $-\ket{h_1}$, which reads
   $\bra{b}\ket{h_1}=0$ for any vector $\ket{h_1}$ in $\{\ket{v_1}\}^{\perp}$.
   Therefore, we obtain
   \begin{align}
    \ketbra{v_1} = \ketbra{b}\label{bv1}
   \end{align}
   and $\bra{a}\ket{u_{i_0}}=0$ since $\ket{b}\neq 0$.

      By taking a unitary operator $U_0$ such that $U_0\ket{a} = \ket{b}$, equation \eqref{nd1} reads
   \begin{align}
    \bra{b}\rho\ket{b} = \bra{a}\sigma\ket{a}.\label{bbaa}
   \end{align}
   Because $q_1 \geq \bra{a}\sigma\ket{a}$ by the definition of $q_1$, equations \eqref{bv1} and \eqref{bbaa} imply $q_1\geq p_1$.
   Therefore, we have $q_1 = p_1$ since we have assumed $p_1 \geq q_1$.
   Then, we also obtain $\ketbra{u_1} = \ketbra{a}$ similarly.

   Taking $\gamma = p_2$, we have
   \begin{align}
    (p_2\mathbb{1} - \rho)\ket{v_2} = 0.\label{v2egn}
   \end{align}
   From the assumption, we have $p_2 \mathbb{1} - \sigma \neq r \ketbra{a}$ for any real number $r$. Thus, there exists $i \geq 2$ with $p_2 - q_{i} \neq 0$ because $\ketbra{u_1} = \ketbra{a}$.
   We take a unitary $U_{i}:=\ketbra{v_2}{a} + \ketbra{b}{u_i} + \sum_{j\geq 3}| v_j\rangle \langle u_{\tilde{j}} | $, where
   \begin{align}
    \tilde{j} :=
    \left\{
\begin{array}{l}
 \displaystyle
  j \quad (j \neq i)\\
  2 \quad (j = i).
\end{array}
\right.
   \end{align}
   Since $U_i$ satisfies $U_i\ket{a}=\ket{v_2}$ by definition, equation (\ref{nd1}) reads
   \begin{align}
    \bra{b}U_i(p_2\mathbb{1} - \sigma)U_i^{\dagger}\ket{b}
    = \bra{v_2}(p_2\mathbb{1} - \rho)\ket{v_2} = 0\label{p2v2}
   \end{align}
   because of \eqref{v2egn}.
   However, we have
   \begin{align}
    \bra{b}U_i(p_2\mathbb{1} - \sigma)U_i^{\dagger}\ket{b}
    = \sum_{j} (p_2 - q_j) |\bra{b}U_i\ket{u_j}|^2
    = p_2 - q_i
   \end{align}
   from the definition of $U_i$ and \eqref{bv1}.
   Thus, combined with \eqref{p2v2}, we obtain $p_2 - q_i = 0$, which contradicts $p_2 - q_i \neq 0$.
   In conclusion, \eqref{rs} is the only possibility.
  \end{proof}


\begin{thebibliography}{99}
\bibitem{Baumer} E. B\"aumer, M. Lostaglio, M. Perarnau-Llobet, and R. Sampaio, {\it Fluctuating work in coherent quantum systems: Proposals and limitations}, arXiv:1805.10096v2 (2018).
\bibitem{CHT} M. Campisi, P. H\"anggi, and P. Talkner, {\it Colloquium: Quantum fluctuation relations: Foundations and applications}, Rev. Mod. Phys. {\bf 83}, 771 (2011).
\bibitem{Kurchan} J. Kurchan, {\it A quantum fluctuation theorem}, arXiv:cond-mat/0007360 (2000).
\bibitem{Tasaki} H. Tasaki, {\it Jarzynski relation for quantum systems and some applications}, arXiv:cond-mat/0009244 (2000).
\bibitem{TLH} P. Talkner, E. Lutz, and P. H\"anggi, {\it Fluctuation theorems: Work is not an observable}, Phys. Rev. E {\bf 75}, 050102 (2007).
\bibitem{Jarzynski} C. Jarzynski, {\it Nonequilibrium Equality for Free Energy Differences}, Phys. Rev. Lett. {\bf 78}, 2690  (1997).
\bibitem{TH2007} P. Talkner and P. H\"anggi, {\it The Tasaki-Crooks quantum fluctuation theorem}, J. Phys. A {\bf 40}, F569 (2007).
\bibitem{Crooks} G. E. Crooks, {\it Entropy production fluctuation theorem and the nonequilibrium work relation for free energy differences}, Phys. Rev. E {\bf 60}, 2721 (1999).
\bibitem{Messiah} A. Messiah, {\it Quantum Mechanics} (North Holland Publishing Company, Amsterdam 1961).
\bibitem{VWT} B. P. Venkatesh, G. Watanabe, and P. Talkner, {\it Transient quantum fluctuation theorems and generalized measurements}, New J. Phys. {\bf 16}, 015032 (2014).
\bibitem{Perarnau-Llobet} M. Perarnau-Llobet, E. B\"aumer, K. V. Hovhannisyan, M. Huber, and A. Acin, {\it No-go Theorem for the Charcterization of Work Fluctuations in Coherent Quantum Systems}, Phys. Rev. Lett. {\bf 118}, 070601 (2017).
\bibitem{Pusz} W. Pusz and S. L. Woronowicz, {\it Passive states and KMS states for general quantum systems}, Commun. Math. Phys. {\bf 58}, 273 (1978).
\bibitem{Alicki} R. Alicki, {\it The quantum open system as a model of the heat engine}, J. Phys. A {\bf 12}, 103 (1979).
 \bibitem{Kosloff} R. Kosloff and Y. Rezek {\it The quantum harmonic Otto cycle}, Entropy {\bf 19}, 136 (2017).
\bibitem{deChiara} G. De Chiara, A. J. Roncaglia , and J. P. Paz, {\it Measuring work and heat in ultracold quantum gases}, New J. Phys. {\bf 17}, 035004 (2015).
\bibitem{HayTaj} M. Hayashi and H. Tajima, {\it Measurement-based formulation of quantum heat engines}, Phys. Rev. A {\bf 95}, 032132 (2017).
\bibitem{TH2016} P. Talkner and P. H\"anggi, {\it Aspects of quantum work}, Phys. Rev. E {\bf 93}, 022131 (2016).
\bibitem{Cerisola} F. Cerisola, Y. Margalit, S. Machluf, A. J. Roncaglia, J. P. Paz, and R. Folman, {\it Using a quantum work meter to test non-equilibrium fluctuation theorems}, Nature Commun.  {\bf 8}, 1241 (2017).
\bibitem{Haag} R. Haag and D. Kastler, {\it An algebraic approach to quantum field theory}, J. Math. Phys. {\bf 5}, 848 (1964).
\bibitem{Kraus} K. Kraus, {\it States, Effects and Operations}, Lecture Notes in Physics {\bf 190} (Springer, Berlin 1983).
\bibitem{Schatten} R. Schatten, {\it A Theory of Cross-Spaces} (Princeton University Press, Princeton, 1950).
 \bibitem{DL} E. B. Davies, and J. T. Lewis, {\it An operational approach to quantum probability}, Comm. Math. Phys. {\bf 17}, 239 (1970).
 \bibitem{Ozawa} M. Ozawa, {\it Quantum measuring processes of continuous variables}, J. Math. Phys. {\bf 25}, 79 (1984).
 \bibitem{Hayashi} M. Hayashi, {\it Quantum Information Theory: Mathematical Foundation}, 2nd ed. (Springer, Berlin, 2017).
\bibitem{Sagawa} T. Sagawa and M. Ueda, {\it Generalized Jarzynski Equality Under Nonequilibrium Feedback Control}, Phys. Rev. Lett. {\bf 104}, 090602 (2010).
\bibitem{Morikuni} Y. Morikuni and H. Tasaki, {\it Quantum Jarzynski-Sagawa-Ueda relation}, J. Stat. Phys. {\bf 143}, 1 (2011).
 \bibitem{WVTCH} G. Watanabe, B. P. Venkatesh, P. Talkner, M. Campisi, and P. H\"anggi, {\it Quantum fluctuation theorems and generalized measurements during the force protocol}, Phys. Rev. E {\bf 89}, 032114 (2014).
 \bibitem{Rastegin13} A. E. Rastegin, {\it Non-equilibrium equalities with unital quantum channels}, J. Stat. Mech.: Theor. Exp. {\bf 2013}, P06016 (2013).
 \bibitem{Smith} A. Smith, Y. Lu, S. An, X. Zhang, J.-N. Zhang, Z. Gong, H. T. Quan, C. Jarzynski, and K. Kim, {\it Verification of the quantum nonequilibrium work relation in the presence of decoherence}, New J. Phys. {\bf 20}, 013008 (2018).
	 \bibitem{Rastegin18} A.~E.~Rastegin, {\it On quantum fluctuations relations with generalized energy measurements}, Int. J. Theor. Phys. {\bf 57}, 1425 (2018).
\bibitem{TMYH} P. Talkner, M. Morillo, J. Yi, and P. H\"anggi, {\it Statistics of work and fluctuation theorems for
microcanonical initial states}, New J. Phys. {\bf 15}, 095001 (2013).
\bibitem{channel} M. A. Nielsen and I. L. Chuang, {\it Quantum Computation and Quantum Information} (Cambridge University Press, Cambridge, 2000).
\bibitem{DHS} N. Datta, A.S. Holevo, and Y. Suhov, {\it Additivity for transpose depolarizing channels}, Int. J. Quant. Inf. {\bf 04}, 85 (2006).
\bibitem{Gardas} B.~Gardas and S.~Deffner, {\it Quantum fluctuation theorem for error diagnostics in quantum annealers}, Sci. Rep. {\bf 8}, 17191 (2018).
\bibitem{WVT} G. Watanabe, B. P. Venkatesh, and P. Talkner, {\it Generalized energy measurements and modified transient quantum fluctuation theorems}, Phys. Rev. E {\bf 89}, 052116 (2014).
\bibitem{Yosida} K. Yosida, {\it Functional Analysis}, 6th ed. (Springer Verlag, Berlin, 1980).

\end{thebibliography}
\end{document}